[1994/12/01]% LaTeX date must December 1994 or later
\documentclass[final]{article}
 \pdfoutput=1
 
\title{A DP Approach to Hamiltonian Path Problem}
\author{Dmitriy Nuriyev}

\usepackage{amsmath,amsthm,amssymb}
\usepackage{graphicx}
\usepackage{float}
\usepackage[margin=2cm]{geometry}

\fontsize{13}{15}
\selectfont

\chardef\bslash=`\\ % p. 424, TeXbook
\newtheorem{thm}{Theorem}[section]

\newtheorem{lem}[thm]{Lemma}

\theoremstyle{definition}
\newtheorem{defn}{Definition}[section]

\theoremstyle{remark}
\newtheorem{rem}{Remark}[section]

\newcommand{\eval}[2][\right]{\relax
  \ifx#1\right\relax \left.\fi#2#1\rvert}

\begin{document}
\maketitle
%\markboth{}
%{}
%\renewcommand{\sectionmark}[1]{}

\section{Abstract}

A Dynamic Programming based polynomial worst case time and space algorithm is described for computing Hamiltonian Path of a directed graph. Complexity constructive proofs  
along with a tested C++ implementation are provided as well. 
The result is obtained via the use of original colored hypergraph structures in order to maintain 
and update the necessary DP states.

\section{Introduction}

Dynamic Programming hardly requires introduction. Since the term was in introduced by Richard Bellman in 1940's it has countless applications and it's power
allows to have fast algorithms in cases where at a blush there seems to be no way to avoid exponential time. Such examples include graph algorithms:
Bellman-Ford, Floyd-Warshall as well as option pricing, numerical solutions to HJB equations ("backwards in time"), discrete optimal control policies, 
knapsack small block size problem \cite{MT2}, Smith-Waterman local sequence alignment algorithm and many others.

Hamiltonian path (H-path) in directed and undirected graph is one of Karp's 21 NP-complete problems \cite{K}. The question it asks is to find cycle or path
in a given graph which visits every vertex exactly once. Therefore one can see that H-path is the shortest path which visits all nodes in the graph and thus
provides the answer to the transportation problem in question.

David Zuckerman \cite{Z} showed in 1996 that every one of these 21 problems has a constrained optimization 
version that is impossible to approximate within any constant factor unless P = NP, by showing that Karp's approach to reduction generalizes to a specific 
type of approximability reduction. Significant progress has been made in expanding the class of graphs for which polynomial solution does exist: notably 
Ashay Dharwadker \cite{A} discovered an algorithm for a broad class of highly connected graphs.

The intuition behind the presented approach comes from physical objects that seem to have a way to solve some NP-complete problems such as soap bubbles forming (almost) 
minimal surfaces and protein folding satisfying hydrophilic-hydrophobic Boolean conditions. 

Consider a cobweb attached to a tree with one end and being pulled by the other end. One can observe that the lowest (average) strain or highest slack in the cobweb is 
achieved along the longest sequence of web segments. This is of course a Hamiltonian 
path in the graph represented by the cobweb. If one wants to find it then it makes sense to implement some incremental scheme which akin to a difference scheme for 
PDE with boundary condition, computes the strains incrementally starting with one attachment point and moving towards the opposite attachment point.

\section{Definitions}

Denote $G,V(G),E(G)$ - the input directed graph, it's vertices (aka nodes) and edges respectively, $n=|V(G)|$.
Assume $G$ to be a general directed graph which does not have multiple edges between any two vertices and that there are no loop edges.

A given instance of H-path problem is represented by a triple $(G,s,e)$ where $s,e \in V(G)$ and the task is to find path $\pi(s,e) \in V(G)$ starting 
at node $s$ and ending at $e$ so that every vertex in $V(G)$ is visited and only once.

In the following description I shall opt for a less conventional but more compact and practical terminology in describing the main algorithm by
adopting Object-Oriented terminology. The entities used here will be Objects which consist of Attributes - representing data and Methods - representing
functions as well as traditional program functions and variables - these are distinct from object methods and attributes in that they are "static", 
effectively having global scope and accessible from all objects.

If an object $A$ is an attribute of $B$ I typically write $A_{B}$ unless making a statement about any object $A$ or it is clear from the context 
what $B$ is. Also instead of single letter notation I elect mnemonics that are much shorter than full object names but are easier to associate 
with full object names despite sounding weird. 

I start with introducing a concurrent graph traversal method which unlike DFS/BFS (Depth First Search/Breadth First Search) has some useful properties in this context. 
This will become a basis for the workflow of main algorithm.

Denote $P(G)$ family of H-paths in $G$ between nodes $s,e$.

\begin{defn}\label{botmaster}
Let $botmaster$ (C++ code: edg) be an object with the following:
\begin{itemize}
	\item Attributes:
		\begin{itemize}
			\item[(1)] $bots$ - family of all $bot$ objects, defined below.
			\item[(2)] march step $marchStep \in 1..n$
			\item[(3)] $dock$ a Boolean variable associated with pair $(v,marchStep)$, $v \in V(G)$    
		\end{itemize}
	\item Methods:
		\begin{itemize}
			\item[(1)] $march$ for all $bot \in bots$ invokes $advance_{bot}$, defined below.
			\item[(2)] $addDelete(v_{1},v_{2})$ (C++: edg::add\_delete) defined below
		\end{itemize}
\end{itemize}
\end{defn}

\begin{defn}\label{bot}
Let $bot$ be an object invoked by master object $botmaster$, which on march step $marchStep$ conditionally transitions from node $a$ to its neighbor $b$ 
along edge $ab$, denote that fact as $bot(a,s) \rightarrow b$. Bot's only attribute is current node $a \in V(G)$ and it has
\begin{itemize}
\item Methods:
	\begin{itemize}
	\item[i] $advance$: 
		\begin{itemize}
			\item[(1)] $bot(a,s) \rightarrow b \Rightarrow dock(b,s)=F$, $dock(a,s-1)=T$. This means dock acts like mutex for bots - first bot accessing $b$ docks at $b$.
			\item[(2)] $addDelete(a,b)_{botmaster}$ is invoked
			\item[(3)] If $bot(a,s) \rightarrow b_{1}...b_{h}, h>1$, $bot(a,s)$ spawns botlets $bot(b_{1},s+1)...bot(b_{h},s+1)$
		\end{itemize}
	\item[ii] $terminate$ - deletes $bot$ object under the following conditions:
		\begin{itemize}
			\item[(1)] If $\not\exists b: bot(a,s) \rightarrow b$, $bot(a,s)$ is terminated.
			\item[(2)] $bot(a,n-2) \rightarrow e$ and there are no other $v: bot(a,n-2) \rightarrow v$ bot march stops.
		\end{itemize}
	\end{itemize}
\end{itemize}
\end{defn}

\begin{defn}\label{botpath}
Let $hist(bot)$ be the history of nodes visited by $bot$, this includes "genetic memory", i.e. newborn botlet inherits history of it's parent.
\end{defn}

\begin{rem}
There are no more than $n$ bots at any given time.

This follows from correspondence between bots and nodes guaranteed by dock condition \eqref{bot}.
\end{rem}

\begin{rem}\label{bothist}
Let $\pi$ is H-path on $G$. Then there exists a $bot$: $hist(bot)=\pi$.
\end{rem}

\begin{proof}
Assume that $\pi_{1}=(s,...,e_{1}) \subset \pi$, is the longest path traversed by any bot in $|\pi_{1}|$ steps. 
By \eqref{bot}.1 we have a bot docked at $e_{1}$ which will traverse to $e_{2}: (e_{1},e_{2}) \subset \pi$ on step $|\pi_{1}|+1$ thus
violating the assumption of maximality of $\pi_{1}$.
\end{proof}

\begin{rem}\label{botvisits}
Any vertex of $G$ is visited at most $n$ times.

This follows from condition \eqref{bot}.ii.(2)
\end{rem}

\begin{defn}\label{color}
Let
\begin{itemize}
	\item Color $C$ of graph $g$ for node $v \in G$ be an integer corresponding to a class of paths $\{(s...v)\}$ in graph $g$. 
	\item $colors(g)$ be all colors of graph $g$. Also: 
	\item Each color has a unique node it relates to via surjective function: $cono: cono(C) \rightarrow V(g)$ (stands for color nodes: C++ dg::color\_nodes).
	\item Denote $\pi(C)$ set class of paths for color $C$. Note that $|\pi(C)|$ may be exponentially big and therefore this is not explicitly computed by the algorithm.
	\item There is a base color which $s$ is painted in - $emptycolor$
\end{itemize}

\end{defn}

\begin{defn}\label{cohi}
Let,
\begin{itemize}
	\item $cohi$ (color hierarchy, C++ color\_hierarchy) be a directed graph of colors, where edge $c_{1}c_{2}$ denotes that every path over nodes of $g$ in $c_{1}$ is a 
	subpath of some $g$-path in $c_{2}$, has form $\{(c,S_{c}),S_{c} \subset colors(g) \}$
	\item $colors(cohi) \equiv V(cohi)$ returns colors of $cohi$
	\item $suco(C)$ is a global function (C++ sub\_colors) which returns all colors in $cohi$ which have path to $C$, i.e sub-colors of $C$,
	also has form 
	\[
	\{(c,S_{c}), S_{c} \subset \bigcup\limits_{v \in V(G),slack=1..n-1} colors(g_{pagra(v,slack)}) \}
	\]
	\item $noco(v)$ (node colors, C++: dg::node\_colors ) is colors over node $v \in V(g)$, has form $\{(v,S_{v}),S_{v} \subset colors(g)\}$ 
	\item $cn(v)$ (colors for node, C++: dg::cn ) is all colors $C$ over $v$ such that $cono(C)=v$
	\item $top(cohi), base(cohi)$ are respectively first (no ancestors) and last (no descendants) colors in $cohi$ graph.
\end{itemize}
\end{defn}

\begin{lem}\label{cncard}
$\forall v \in g, |cn(v)|=O(n)$
\end{lem}
\begin{proof}
Be remark \eqref{botvisits} we have $O(n)$ bit visits to $v$. Each visit performs merge for given $pagra(v,slack)$ which generates new $C: cono(C)=v$. 
\end{proof}

\begin{defn}\label{inactivecolors}
Color $C$ is called inactive: $inactive(C)$ iff 
\begin{itemize}
\item $cono(C)=\emptyset$
\item or $cono(C) \not\in V(g)$
\item or $noco(cono(C))=\emptyset$
\item or $C \not\in top(cohi) \wedge cohi(C)=\emptyset$
\end{itemize}
\end{defn}

\begin{defn}\label{depcolors}
Let $S \in colors(cohi), a \in V(g)$ $dep(S|cn(a))$ is a subset of $S$ such that either 
\[
C \in S, \not\exists \pi(C,top(cohi)) \in cohi: cn(a) \cap \pi(C,top(cohi))=\emptyset  
\]
Or
\[
C \in S, \not\exists \pi(C,base(cohi)) \in cohi: cn(a) \cap \pi(C,base(cohi))=\emptyset  
\]
\end{defn}

\begin{defn}\label{pagra}
For all $a \in V(G), slack \in 1...n$ call path graph $pagra(a,slack)$ (C++: dg) an object which includes all paths $(s...a) \subset G$ of length $slack$ and has:
\begin{itemize}
\item Attributes:
	\begin{itemize}
		\item[(1)]  Directed graph $g$ with nodes in $V(g)$
		\item[(2)]  Node colors $noco(V(g))=\{(n_{k} \rightarrow \{c_{k}^{1},..,c_{k}^{h}\})\}$, $noco$ is a one-to-many map $V \rightarrow C$
		\item[(3)]  Color nodes $cono(C)$ 
		\item[(4)]  Color hierarchy $cohi(g)$ 
	\end{itemize}
\item Methods:
	\begin{itemize}
		\item[(1)]  $addSlack(b), b \not\in V(pagra(a,slack))$, does:
			\begin{itemize}
			\item[i]  adds edge $ab$ to $g_{s}, \forall m$
			\item[ii] invokes $paint$ which performs: 
				\begin{itemize}
					\item[(1)] allocates new C 
					\item[(2)] $\forall v \in V(g), noco(v):=(noco(v),C)$
					\item[(3)] $cono(C):=b$
					\item[(4)] $cohi(C):=top(cohi)$ 
					\item[(4)] $suco(C):=colors(cohi)$
				\end{itemize}
			\end{itemize}
		\item[(2)]  Remove node (C++: dg::rm\_node) $reno(a)$:
			\begin{itemize}
				\item[i] deletes $noco(a)$
				\item[ii] executes $bleach(cn[a])$ which:
					\begin{itemize}
					\item[(1)] computes $D=\{dep(cn[v]|cn[a])| \forall v \in V(g)\}$
					\item[(2)] $\forall C \in D$ remove $C$ from all attributes of $pagra(a,slack)$
					\item[(3)] ensures $\forall C \in colors(cohi), inactive(C)=F$
					\item[(4)] every color $C \in cohi$ has a direct predecessor or $bleach$ removes such $C$
					\item[(5)] if $noco(v)=\emptyset$, remove $v$ from $g$ $\forall v \in V(g)$
					\end{itemize}
			\end{itemize}
		\item[(3)]  Merge $pagra(a)+pagra(b)$: unions are taken for $cohi,cono,noco$, i.e. it has form: 
		\[
		\{(a,S_{a})\}+\{(b,S_{b})\}=\{a,b: a = b \Rightarrow (a,S_{a} \cup S_{b}), otherwise (a,S_{a}),(b,S_{b}) \} 
		\]
	\end{itemize}
\end{itemize}

\end{defn}

\begin{defn}
Denote 
\[
sucoDFS(a,X,suco,dir), a \in colors(cohi), X \subset colors(cohi), dir \in \{up,down\} 
\]
a $pagra$ method (C++: dg::has\_path\_to\_top()) which traverses $cohi$ starting from color $a$, bypassing colors $X$ (C++: dg::xcolors) and 
uses $suco$ to make transitions between nodes as follows:
\begin{itemize}
\item[(1)] $dir=up$ 
	\begin{itemize}
		\item[(1)] run regular DFS against graph $cohi$ starting at color $a$ 
		\item[(2)] when next node is $c \in X$ return to the next available sibling node
		\item[(3)] transition from color vertex $a_{1}$ to $a_{2}$ only if $\pi_{1} \subset suco(a_{2})$ where
			$\pi_{1} \equiv (a...a_{1})$ is a sequence of colors - path - in $cohi$ traversed before $a_{2}$.
		\item[(4)] terminate and return path $\pi$ when next node is in $top(cohi)$ 
		\item[(5)] when all colors have been traversed terminate and return $\emptyset$ 
	\end{itemize}
\item[(2)] $dir=down$
	\begin{itemize}
		\item[(1)] run regular DFS against graph $inv(cohi)$ starting at color $a$ 
		\item[(2)] when next node is $c \in X$ return to the next available sibling node
		\item[(3)] transition from color vertex $a_{1}$ to $a_{2}$ only if $\forall c \in \pi_{1} a_{2} \in suco(v)$
		\item[(4)] terminate and return path $\pi$ when next node is in $base(cohi)$ 
		\item[(5)] when all colors have been traversed terminate and return $\emptyset$ 
	\end{itemize}
\end{itemize}
\end{defn}

\begin{rem}
By definition of $sucoDFS$ we have: 
\[
sucoDFS(a,X,suco,up)=\emptyset \vee sucoDFS(a,X,suco,down)=\emptyset \Rightarrow dep(a|X)=\{a\}
\]
\end{rem}

\begin{defn}\label{inv}
Given directed graph $G$ $inv(G)$ inverts arc orientation.
\end{defn}

\begin{defn}\label{slacks}

Slacks $slacks(a)$ is an object, effectively a wrapper for $pagra(a,slack)$ with:

\begin{itemize}
\item Attributes:
	\begin{itemize}
	\item[(1)] A single node $v \in V(G)$
	\item[(2)] A set of slacks $S$
	\item[(3)] Path graphs $pagra(a,slack), slack \in S$
	\end{itemize}

\item Methods
	\begin{itemize}
	\item[(1)] Add slack - $addSlack(b), b \not\in V(g_{pagra(a,slack)})$, invokes $addSlack(b)_{pagra(b,slack)}, \forall slack$
	\item[(2)] Remove node: $reno(a)$ invokes $reno(a)_{pagra(b,slack)}, \forall slack$
	\item[(3)] Merge $slacks(a)+slacks(b)$: invokes $pagra(a,slack)+pagra(b,slack), \forall slack \in S$
	\end{itemize}
\end{itemize}
\end{defn}

\begin{defn}\label{slacks}
$addDelete(a,b)_{botmaster}$ is a method which performs the following steps:
\begin{itemize}
\item $reno_{slacks(a)}(b)$
\item $addSlack_{slacks(a)}(b)$
\item Merge $slacks$: $slacks(b):=slacks(a)+slacks(b)$
\end{itemize}
\end{defn}

\begin{figure}[H]
\centering
\label{fig:obj}
\caption{Main objects}
\includegraphics[width=130mm, height=100mm]{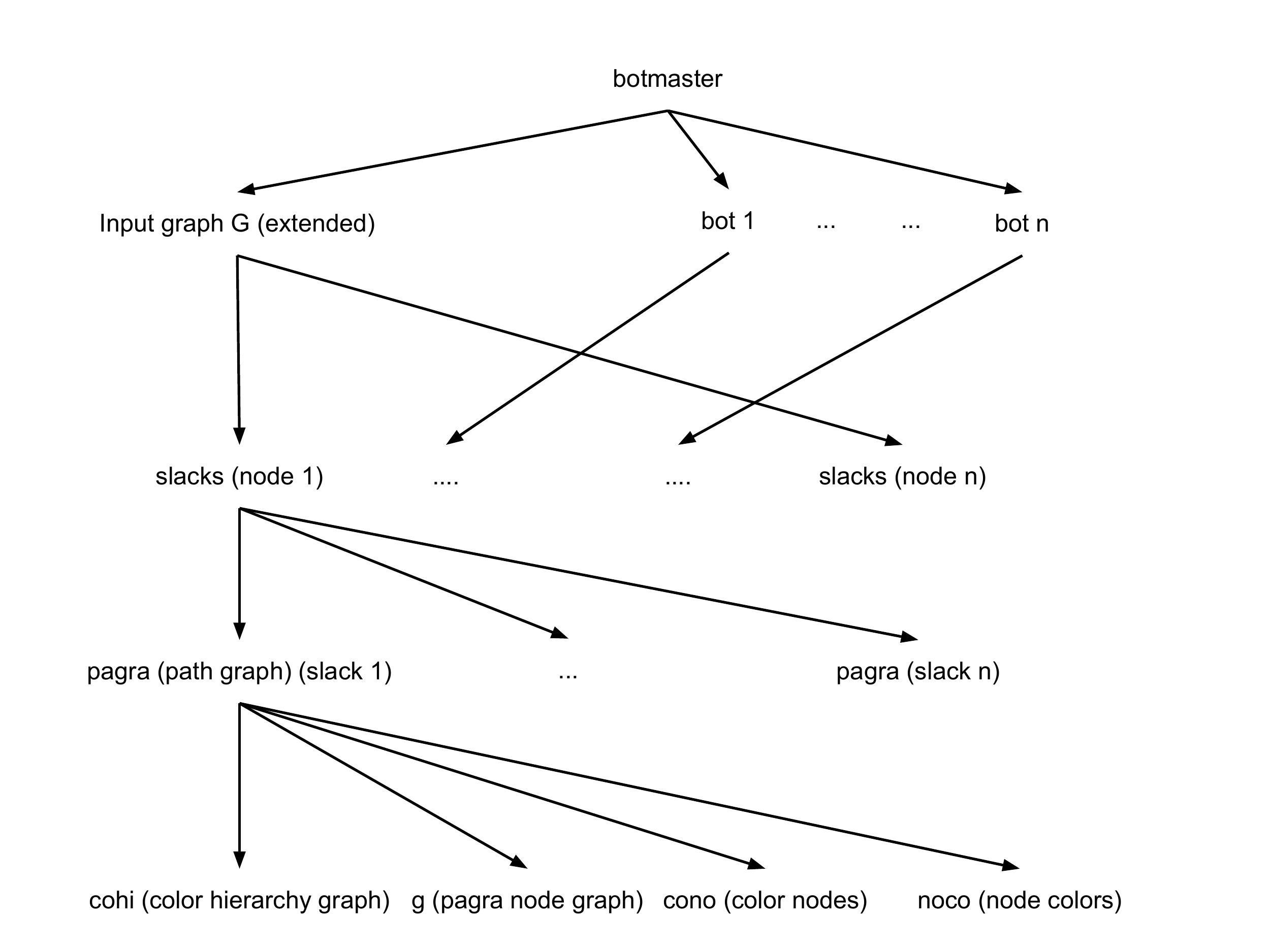}
\end{figure}

\begin{defn}
For a connected sequence of nodes $\pi$ $dup(\pi)$ is a set of repeated nodes.
\end{defn}

\begin{defn}
Denote $P'(G)$ class $sucoDFS$-traversible paths $\forall s,e \in V(cohi),X=\emptyset$.
\end{defn}

\section{Explanation}

The idea here is to categorize families of paths without explicitly describing each one of them. $Slacks$ object for a given node encodes
paths to a given node over $G$ and groups them by length (slack). For each given slack $pagra$ object is the encoding of paths given by
all paths traversible in it's own copy of graph $g$. Graph $g$ by itself is not enough to keep track of those paths due to the way it has to be updated:
via $merge$ and $reno$ operation. $reno$ operation ensures paths don't  run over the same node twice since $addDelete$ removes each node $v$ from $g$
when $bot(u,marchStep) \rightarrow v$. After that, $merge$ executes $pagra(u,slack)+pagra(v,slack)$ (if $pagra(v,slack) \neq \emptyset$).
If $g_{1}=g_{pagra(u,slack)}$ and $g_{2}=g_{pagra(v,slack)}$ both contain some node $w$ then the resulting $g$ will lose track of paths through $w$
resulting in "synthetic path" which is not a path in $G$ for a given $slack$.
To preserve actual paths, concept of color is introduced. However just the color alone is not enough since we have cases when we start with
some color, for example, "red" in $slack_{0}$ which then "splits" into $slack_{1}$ and $slack_{2}$ each one carrying it's own version of "red" which then
gets modified by their own $reno$ operations. After a while slacks combine again but now red two distinct path classes. 
We avoid exponential explosion in number of colors by using $cohi$, which associates each copy of $red$ with some own unique color
added in $slacks_{1}$, for example, "blue" and in $slacks_{2}$ - "brown". These new colors are added as part of $addDelete$ method in bot march and so their number is
"small" as shown below. When $slack_{1}+slack_{2}$ is computed, $suco$ includes $(brown,(red,...))$, $(blue,(red,...))$. Method $sudoDFS$ uses $suco$
to avoid "synthetic path" as shown below.

\section{Path existence}

\begin{lem}\label{numcolors}
Total number of colors is $|colors(suco)|=O(n^3)$. 
\end{lem}
\begin{proof}
Indeed, a color is assigned by $paint$ method which is called for:

\begin{itemize}
\item every vertex in $v \in V(G)$ - $O(n)$
\item every bot visit to $v$ - $O(n)$ by Remark \eqref{botvisits}
\item every $pagra(v,slack)$ - $O(n)$
\end{itemize}

\end{proof}

\begin{lem}\label{mergedgs}
Let $\pi \in P'(cohi_{1}+cohi_{2})$ then 
\[
\not\exists \pi _{1} \subset P(cohi_{1}), \pi _{2} \subset P'(cohi_{2}): \pi=(\pi_{1},\pi_{2})
\]
i.e. $sucoDFS$ does not return $\pi$ that starts in $cohi_{1}$ and continue in $cohi_{2}$.
\end{lem}
\begin{proof}
Assume such path $\pi$ existed, then we would have 
\[
c_{1} \in colors(cohi_{1}) \backslash colors(cohi_{2}), c_{2} \in colors(cohi_{2}) \backslash colors(cohi_{1}): c_{1}, c_{2} \in \pi
\]
See Figure 2 for illustration.
But then we must have $c_{1} \in suco(c_{2})$ by rules \eqref{depcolors}.1-2 which is impossible because $c_{1}$ is not even in $colors(cohi_{2})$

\begin{figure}[H]
\centering
\label{imp}
\caption{Impossible synthetic path on merger}
\includegraphics[width=130mm, height=90mm]{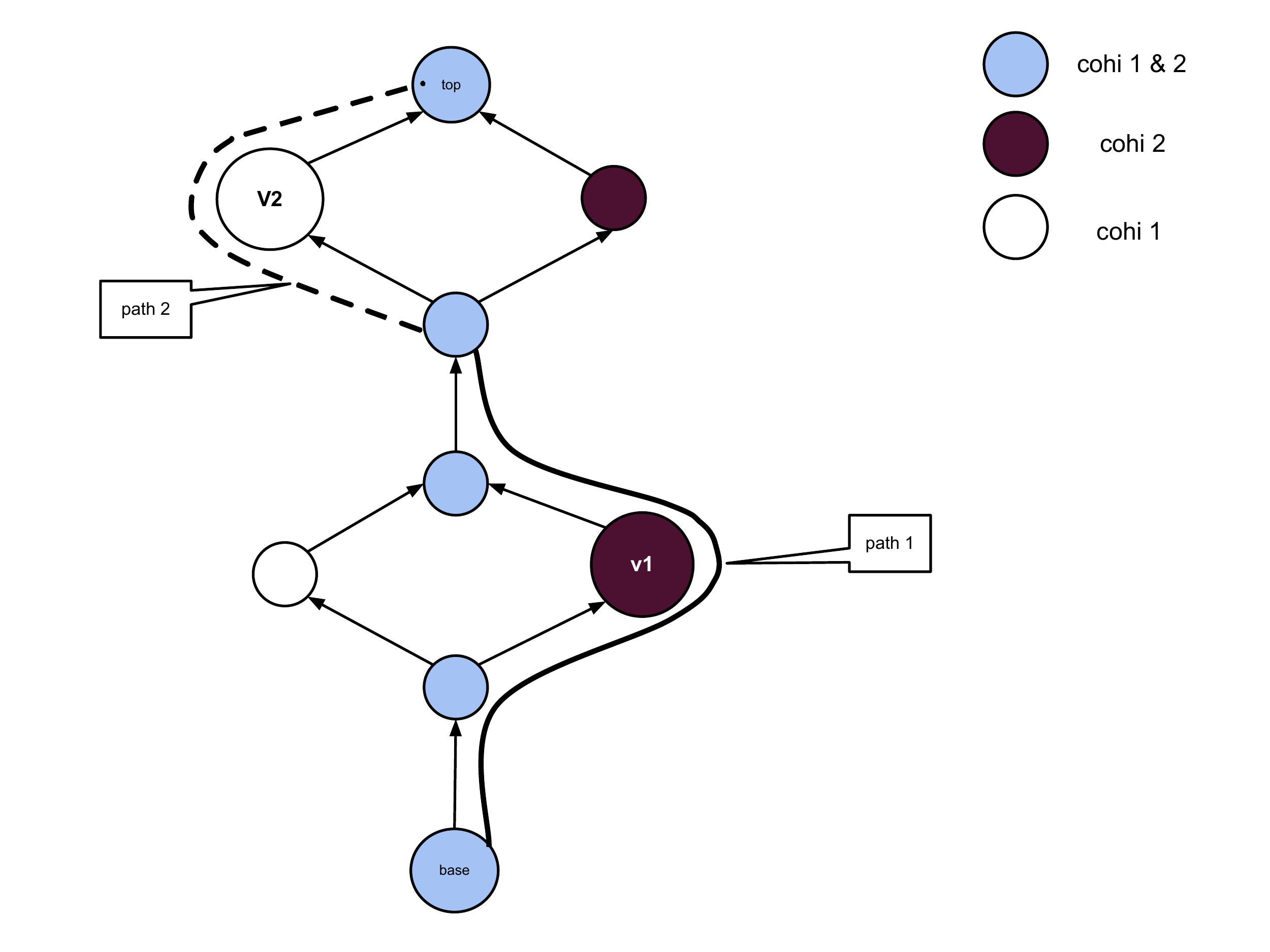}
\end{figure}

\end{proof}

\begin{lem}\label{sidenode}
Let $\gamma$ be a $sucoDFS$-traversible path, $\pi = cono(\gamma), u \in \pi$ and importantly $v \in g$ but $v \not\in \pi$, then 
\[
dep(cn(u)|cn(v))=\emptyset, \forall u \in \pi
\]
\end{lem}
\begin{proof}
By definition of $dep$, $\forall u \in \pi$, $cn(u)$ has path to $top(cohi)$ and that path is $\gamma$ and $\gamma \cap cn(v) = \emptyset$ 
therefore $dep(cn(u)|cn(v))=\emptyset$.
\end{proof}

\begin{lem}\label{cohi_to_G}
Let
\[
\pi=cono(sucoDFS(top(cohi_{pagra(b,j)}),\emptyset,suco,down))
\]
then $dup(\pi)=\emptyset, |\pi|=j$.
\end{lem}
\begin{proof}
Assume there was an intersection $dup(\pi) \neq \emptyset$ then we have $c_{1},c_{2}: cono(c_{1})=cono(c_{2})$ however if $c_{2}$ was traversed after $c_{1}$ then 
by \eqref{depcolors}.1 observe that $c_{1} \in suco(c_{2})$ but by $suco$ construction this is impossible if $cono(c_{1})=cono(c_{2})$.

$|\pi|=j$ folows from $bleach$ method definition: every color $C \in cohi$ has a direct predecessor or $bleach$ removes $C$, also $paint$ method and $cono$
surjectivity ensures that any color path in $cohi_{pagra(b,j)}$ from $base$ to $top$ is $j$ long.

\end{proof}

Lemma \eqref{cohi_to_G} allows us to recover final H-path from $cohi_{pagra(e,n-1)}$ using standard $sucoDFS$ function. That is if such path 
is included in $cohi_{pagra(e,n-1)}$, let's show this is indeed the case:

\begin{thm}\label{hpath}
If $P(G) \neq \emptyset$ then $\exists \pi \in P(G)$:
\[
 \pi=cono(sucoDFS(top(cohi_{pagra(e,n-1)}),\emptyset,suco,down))
\]
\end{thm}
\begin{proof}
Using induction by march step, this trivially holds at node $s$, $marchStep=1$, assume at $marchStep=k$ we have: 
\[
\gamma_{j} \equiv sucoDFS(top(cohi_{pagra(v_{j},j)},\emptyset,suco,down) 
\]
Let $\pi_{j} \equiv cono(\gamma_{j})$.
When $bot(v_{j},j) \rightarrow v_{j+1}$ we can show that, for $\pi_{j+1} \equiv (\pi_{j},v_{j+1})$
\[
\pi_{j+1} = cono(sucoDFS(top(cohi_{pagra(v_{j+1},j+1)},\emptyset,suco,down))
\]

Indeed $bot(v_{j},j) \rightarrow v_{j+1}$ invokes $addDelete$ which invokes methods $merge$ and $reno$, $merge$ does not affect
$\pi_{j}$ by lemma \eqref{mergedgs} while if $v_{j+1} \in V(g_{pagra(v_{j},j)})$ $reno$ will invoke $bleach(dep(cn(v)|cn(v_{j+1}))),\forall v \in \pi_{j}$ 
but since $\pi_{j}$ is $sucoDFS$ traversable and $v_{j+1} \not\in \pi_{j}$ therefore by lemma \eqref{sidenode}:
\[
colors(\gamma_{j}) \not\subset dep(cn(v)|cn(v_{j+1}), \forall v \in \pi_{j}
\]
therefore $\gamma_{j}$ and therefore $\pi_{j}$ is not affected by $bot(v_{j},j) \rightarrow v_{j+1}$. 
\end{proof}

\section{Algorithm}

\begin{itemize}
\item Start with one bot which current node set to $s$
\item $slacks$ object and the associated $pagra$ objects are initialized
	\begin{itemize}
		\item $slacks$ consists of just one $pagra$
		\item $g_{pagra}$ consists of only one node $s$
		\item $cohi$ consists of just $emptycolor$
		\item $noco$ is just $(s,emptycolor)$ pair
		\item $cono$ is $(emptycolor,s)$
	\end{itemize}
\item Bot advances $s \rightarrow v_{1}$ according to it's rules, and invokes $addDelete$ function with updates $slacks(v_{1})$
\item Steps above are repeated, until $pagra(e,n-1) \in slacks(e)$, i.e. we have a path graph with slack equal to $|V(G)|- 1$.
\item Extract the H-path by computing 
\[
cono(sucoDFS(top(cohi_{pagra(e,n-1)}),\emptyset,suco,down))
\]
\end{itemize}

\section{Complexity}

Time complexity is based on that:

\begin{itemize}
\item there are $O(n)$ $G$ nodes 
\item visited by bots at most $O(n)$ times
\item each visit requires $runo$ call for each of $O(n)$ graphs $pagra(v,slack)$ 
\item $reno$, requires $sucoDFS$ call for each color $a$ of at most $O(n)$ sized (lemma \eqref{cncard}) $cn(v)$ color set 
\item $sucoDFS(a,cn(b),suco,dir)$ complexity is $O(n^3)$ since:
	\begin{itemize}
		\item any $sucoDFS$-traversible path in $cohi$: $\pi(base(cohi),top(cohi))$ is $O(n)$ long by lemma \eqref{cohi_to_G}
		\item set $X=cn(b)$ cardinality is at most $O(n)$ by lemma \eqref{cncard} and therefore $DFS$ requires at most $O(n)$ rollbacks when $X$ color is encountered 
		\item validating conditions \eqref{depcolors}.1,2 on every step costs $O(n)$
	\end{itemize}
\item which is executed for each $O(n)$ nodes $v \in g$ to decide if the node is bleached 
\end{itemize}

Space complexity is based on that:

\begin{itemize}
\item there are $O(n)$ $G$ nodes 
\item each node has $O(n)$ $pagra$ objects
\item each $pagra$ object has $cohi$ as biggest attribute 
\item each $cohi$ attribute has at most $O(n^3)$ colors
\end{itemize}

Thus total current worst case time complexity is $O(n^8)$ and space complexity $O(n^5)$. Average case complexity is typically much lower because in practice $cn(v)$ is $O(1)$ sized and average number
of node bot revists is also $O(1)$.

$sucoDFS(a,cn(b),suco,dir)$ complexity can be lowered by $O(n)$ if we eliminate $cn(b)$ from $cohi$ first which is done once for all colors $a$. 
Than for each $a$ we only have to run $sucoDFS(a,\emptyset,suco,dir)$ in $O(n)$.

Complexity can likely be further lowered by an additional $O(n)$ by consolidating all $cohi$ per $G$ node and across $pagra$ of different slack levels. 
Slack values can be  extracted from $cohi$ directly when necessary.

Additionally $dep(cn(b)|cn(a))$ can possibly be found faster if we start with an  encoding of all $sucoDFS$-traversible paths $P$ to top.
Then on every update of $cohi$ instead or rerunning $sucoDFS$ we could just update $P$.

\section{Implementation and testing}

The algorithm was tested on 10000 randomly generated graphs each with $17$ nodes and $3$ outbound edges per node. Test graphs were generated using 
function $dg::gen\_graph$ which is part of the source code and is as follows;

\begin{itemize}
\item Generate a random H-path first. Sample without replacement from a discrete uniform random distribution over $1...n: U(1...n)$.
\item For each node on selected H-path generate given number of edges $\delta$ by choosing vertex 1 sequentially from the path and then choosing vertex 2 
by sampling without replacement from $U(1...n)$ $\delta$ times.
\end{itemize}

Please note that performance tuning was not a priority therefore the author is aware of a number of inefficiencies. 

The algorithm is available for download at 

https://docs.google.com/folder/d/0B3s6PXhKJO6HWnE4c0VZbVJodDQ/edit . 

It is free to use for academic and educational purposes use but requires a license for commercial and government use.
Potential commercial and government users should note the algorithm has a Patent Pending status with USPTO.

\section{Appendix}

\begin{figure}[H]
\caption{Source Code}
\centering
\includegraphics[width=210mm, height=240mm,page=1]{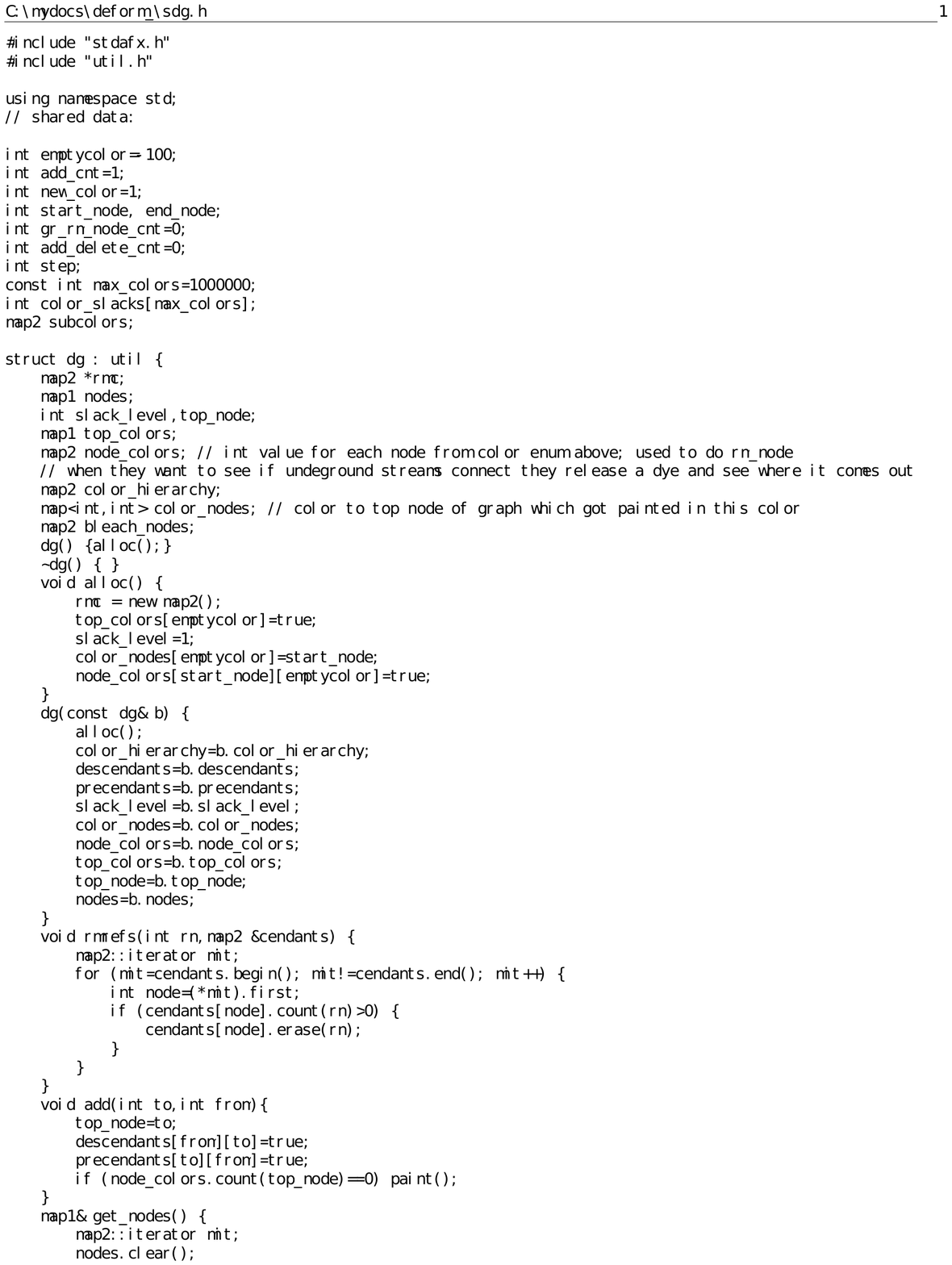}
\end{figure}
\begin{figure}[H]
\centering
\includegraphics[width=210mm, height=240mm,page=2]{sdg}
\end{figure}
\begin{figure}[H]
\centering
\includegraphics[width=210mm, height=240mm,page=3]{sdg}
\end{figure}
\begin{figure}[H]
\centering
\includegraphics[width=210mm, height=240mm,page=4]{sdg}
\end{figure}
\begin{figure}[H]
\centering
\includegraphics[width=210mm, height=240mm,page=5]{sdg}
\end{figure}
\begin{figure}[H]
\centering
\includegraphics[width=210mm, height=240mm,page=6]{sdg}
\end{figure}
\begin{figure}[H]
\centering
\includegraphics[width=210mm, height=240mm,page=7]{sdg}
\end{figure}
\begin{figure}[H]
\centering
\includegraphics[width=210mm, height=240mm,page=8]{sdg}
\end{figure}
\begin{figure}[H]
\centering
\includegraphics[width=210mm, height=240mm,page=9]{sdg}
\end{figure}
\begin{figure}[H]
\centering
\includegraphics[width=210mm, height=240mm,page=10]{sdg}
\end{figure}

\begin{figure}[H]
\centering
\includegraphics[width=210mm, height=240mm,page=1]{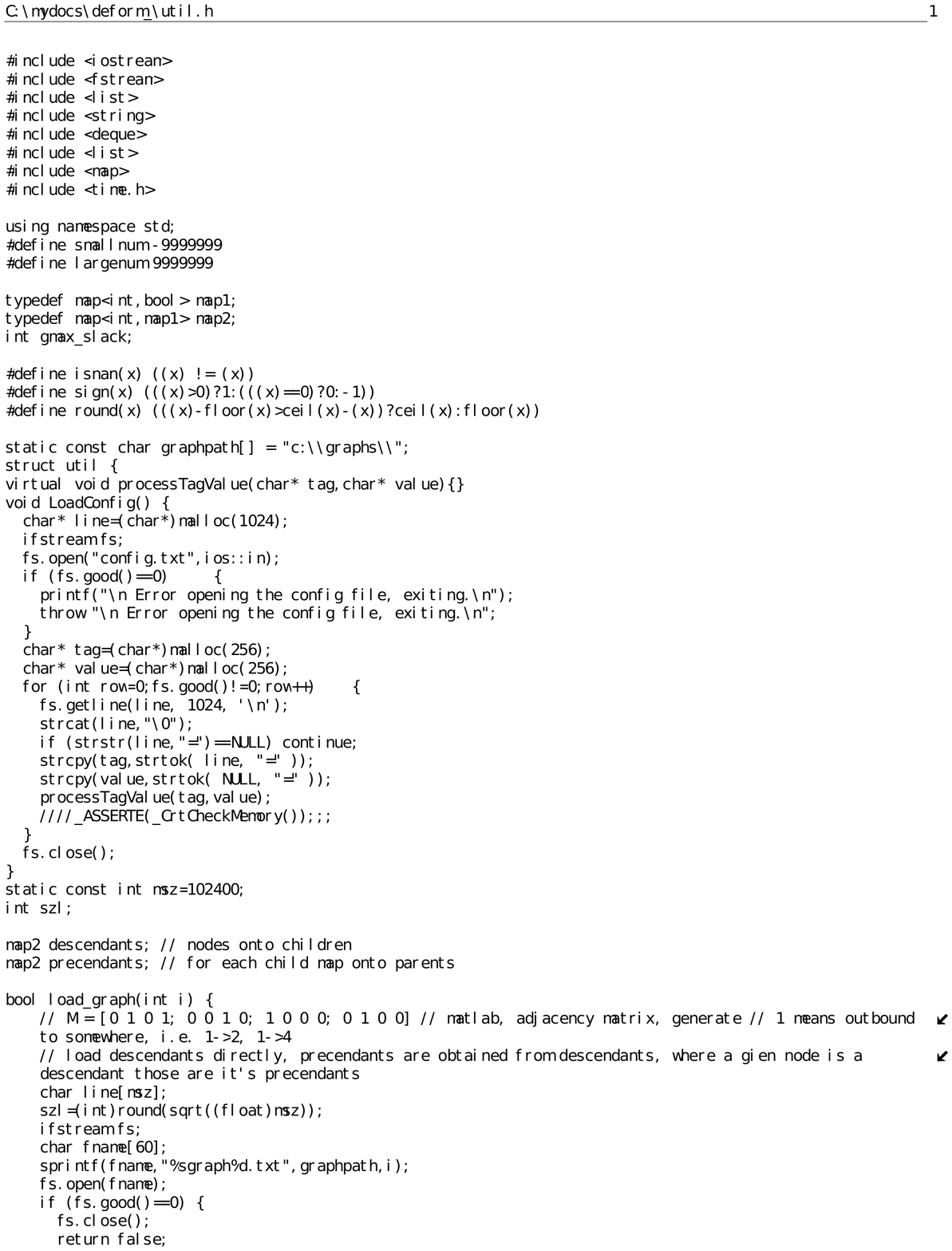}
\end{figure}
\begin{figure}[H]
\centering
\includegraphics[width=210mm, height=240mm,page=2]{util}
\end{figure}
\begin{figure}[H]
\centering
\includegraphics[width=210mm, height=240mm,page=3]{util}
\end{figure}
\begin{figure}[H]
\centering
\includegraphics[width=210mm, height=240mm,page=4]{util}
\end{figure}
\begin{figure}[H]
\centering
\includegraphics[width=210mm, height=240mm,page=5]{util}
\end{figure}
\begin{figure}[H]
\centering
\includegraphics[width=210mm, height=240mm,page=6]{util}
\end{figure}

\begin{figure}[H]
\centering
\includegraphics[width=210mm, height=240mm]{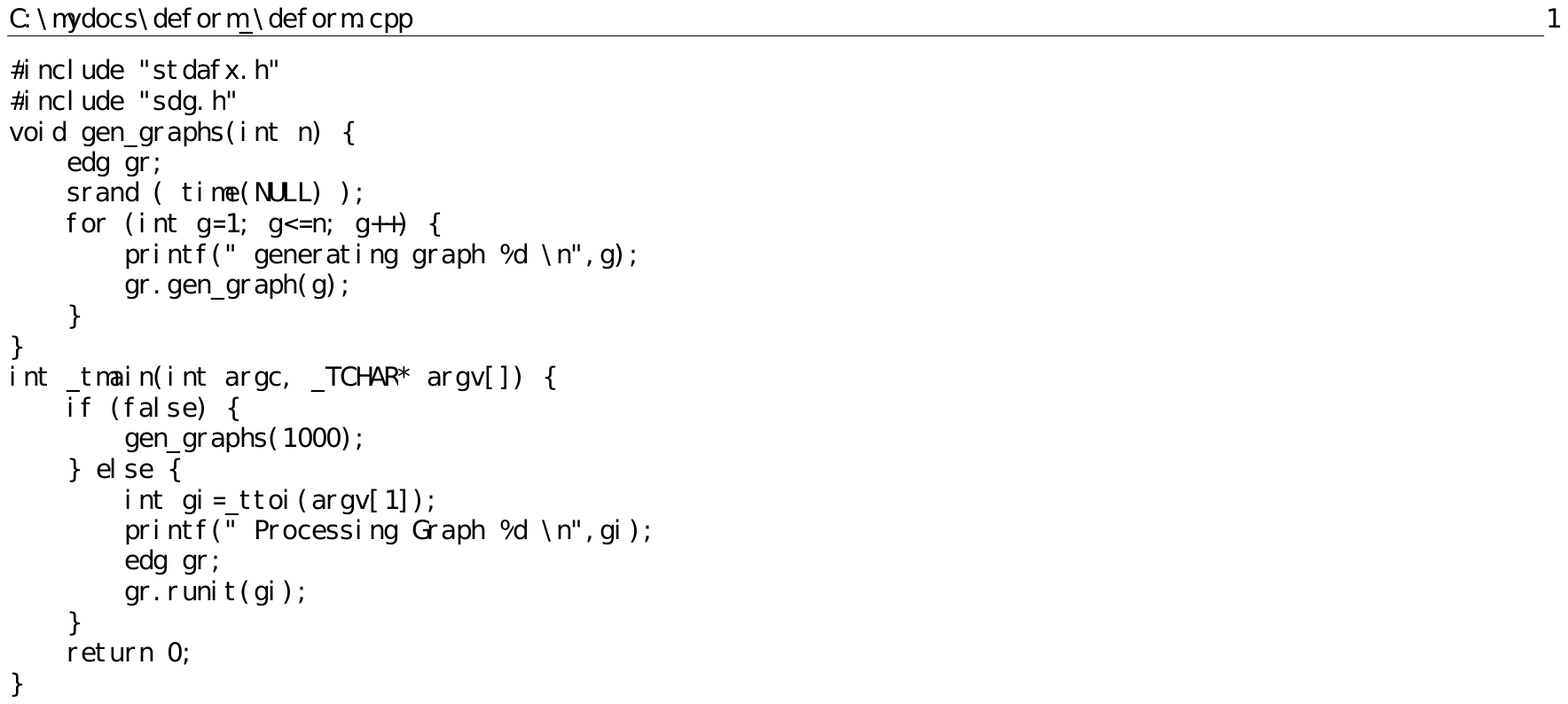}
\end{figure}

\end{document}